\documentclass[a4paper]{article}
\usepackage{amsmath,amssymb,amsfonts,amsthm} 
\usepackage{enumerate}
\usepackage{graphicx} 

\usepackage{wrapfig} \usepackage{url} \usepackage[usenames]{color}
\def\b1{\mbox{\boldmath $1$}}

\allowdisplaybreaks

\theoremstyle{plain} \newtheorem{thm}{\bf Theorem}[section]
 \newtheorem{lem}[thm]{\bf Lemma}
 \newtheorem{defn}[thm]{\bf
  Definition} \theoremstyle{remark}

\makeatletter \@addtoreset{equation}{section} \makeatother \makeatletter
\@addtoreset{thm}{section} \makeatother
\allowdisplaybreaks 

\usepackage[numbers,sort&compress]{natbib} 
\usepackage[bookmarksnumbered, bookmarksopen,
colorlinks,citecolor=blue,linkcolor=blue]{hyperref}

\begin{document}
\date{}

\title{  Optimal investment-reinsurance policy under a long-term perspective}

\author{XiaoXiao Zheng\thanks{School of Mathematical Sciences and LPMC, Nankai University}, Xin Zhang\thanks{School of Mathematical Sciences and LPMC, Nankai University, Tianjin 300071, P.R. China; E-mail: nku.x.zhang@gmail.com} } \maketitle

\noindent{\bf Abstract}
In this paper, we assume an insure is allowed to purchase proportional reinsurance and can invest his or her wealth into the financial market where a savings account, stocks and bonds are available. Different from classical optimal investment and reinsurance problem, this paper studies the insurer's long-term investment decision. Under this setting, our model consider the interest risk and the inflation risk. Specifically, we suppose the interest rate follows a stochastic process, while price index is described by a classical model. By solving Hamilton-Jacobi-Bellman equation, the closed-form expression of the optimal policy is obtained. Further, we prove the corresponding verification theorem without the usual Lipschitz condition. In the end, numerical examples are made to illustrate the difference of the optimal polices under Ho-lee model and Vasicek model.
\noindent
\\


\newpage

\section{Introduction}
Economist Merton pioneered in the study of continuous-time portfolio problem (see \citet{merton1969lifetime,merton1971optimum}). Base on his notable works, a large number of investment problems have been studied, and a series of classical papers have came out. \citet{davis1990portfolio} is about optimal consumption and investment decision for an investor who invests his or her wealth into a bank account and a stock, in addition, in this paper the author also think of the fixed percentage transaction costs. Life-cycle model of consumption and portfolio was considered by \citet{cocco2005consumption}. \citet{kraft2005optimal} studied the portfolio problem with stochastic volatility. Here we cannot list every important literatures in this area, but these classical papers constitute the cornerstone of the development of continuous-time finance.

Over the past several years, long-term investing strategy has been a hot issue. Many financial gurus advocate that the investor should practise long-term investment. In the famous book \textbf{Strategic Asset Allocation}, campbell also studies the asset allocation decisions for the investor who want to invest in a long term. For an insurance company, it is important to conduct a prudent investment strategy, that is, they not only should diversify risk by investing their money into different assets, but also should practise long-term investment, and our paper is aiming at designing a  portfolio for insurance company. Under the setting of long-term investing, inflation is a main risk factor, thus we have to take consumer price index(CPI) into account. Different papers provide different models to describe such risk, \citet{pearson1994exploiting} shows that the inflation follows a " mean-reverting square-root " process, while in \citet{munk2004dynamic} and \citet{brennan2002dynamic} the inflation process is given by Ornstein-Uhlenbeck process. In our paper, we will choose the second model to characterize the inflation. On the other hand, the interest risk is another important factor we have to focus. In contrast to most classical investment problems, papers in this area assume interest rate follows a stochastic process. Of all papers concern this subject, \citet{korn2002stochastic} is a classical one. In this instructive paper, a investor invest his or her wealth into a bond and a bank account with stochastic interest rate. Besides, there are some papers also focus on this topic: In \citet{li2009optimal} stochastic interest rate is given by Cox-Ingersoll-Ross(CIR) model and the volatility of the stock is also a CIR process. Another paper is \citet{kraft2009optimal}, considering all common short rate models and stochastic discount.

No doubt, financial market is abundant and insurance company plays an important role. In recent years, many scholars were devoted to the study of the investment for insurance company. From optimal Mean-Variance problem to maximizing utility function problem, more and more elegant results were made. \citet{bai2008dynamic} consider the optimal investment and optimal reinsurance for an insurer under the criterion of mean-variance. In \citet{zhang2012optimal}, the author discuss the problem of optimal proportional reinsurance and investment under the criterion of maximizing utility function on terminal wealth.

As far as I know, however, there are seldom literatures touch on the problem of optimal investment and optimal reinsurance for an insurer under stochastic interest rate. In our paper we consider an insurer invest his or her wealth into financial market in which saving account, stock and bond are available. We aim at maximizing the terminal power utility function. For Ho-lee model and Vasicek model, not only do we obtain the closed-form expression of their optimal policies, but also we compare their result through numerical example. Because in our model  Lipschitz condition and growth condition are not satisfied, we prove the verification theorem through a way which is different from the one used in the standard verification theorem.

The rest of this paper is organized as follows. In section 2, we introduce models and some assumptions. Section 3 formulate the optimization problem and then by solving HJB equation we obtained optimal strategies for both Ho-lee model and Vasicek model. In section 4, verification theorem was given. Section 5 provides numerical comparison and analysis.

\section{Model and Assumption }
Let $(\Omega,\mathcal{F},\mathbb{F}, \mathcal{P})$ be a complete filtered probability space, where  $\mathbb{F}=\{\mathcal{F}_t\}_{t\in [0,T]}$ is a right continuous, $\mathcal{P}$-complete filtration to which all of the processes defined below are adapted. 
\subsection{Price index}
In this paper, we try to solve a long-term (10 years or more) optimal reinsurance and  portfolio choice problem for an insurer. When investing for long-term goals, inflation risk is an important factor that can affect the overall performance of our investment.  In economics, the inflation is essentially the  sustained increase in the general price level of goods and services over a period of time. In fact, the inflation reflects a reduction of the purchasing power since when the general price level rises, each unit of currency buys fewer goods and services. Therefore, the accumulated inflation can lead to huge shrinkage in the wealth of our investment. A chief measure of the inflation is the inflation rate (the annualized percentage change in a general price index over time). 
Recently, many researchers studied the portfolio choice problems under the influence of inflation, see for example: \citet{munk2004dynamic,pearson1994exploiting,brennan2002dynamic}.


As adopted in \citet{munk2004dynamic} and \citet{yao2013markowitz}, the dynamic of the nominal price index of the consumption good in the economy is modeled  by  the following stochastic differential equation (SDE for short):
\begin{eqnarray}
  d\Pi(t)=\Pi(t)\big[I(t)dt+\sigma_0(t)dW_0(t)\big], \mbox{ for all } t\in [0,T], \quad \Pi(0)=\Pi_0,
\end{eqnarray}
where $W_0(t)$ is a one dimensional standard Brownian motion, $\sigma_0(t)$ is the volatility of the price index,  and $I(t)$ is the instantaneous expected inflation rate following an time-dependent Ornstein-Uhlenbeck (O-U for short) process
\begin{eqnarray}
dI(t)=\beta(t)\big[\alpha(t)-I(t)\big]dt + \bar \sigma_0(t)dW_0(t),
\end{eqnarray}
where $\alpha(t)$ describes the long-run mean of the inflation rate, $\beta(t)$ descries the degree of mean-reversion and $ \bar \sigma_0(t)$ reflects the volatility of the inflation rate.  Moreover we suppose that  $\sigma_0(t)$, $\alpha(t)$, $\beta(t)$, $ \bar \sigma_0(t)$ are deterministic and continuous function of time $t$.  In our model, we adopt time-dependent O-U process instead of the one introduced in \citet{munk2004dynamic} or \citet{yao2013markowitz}.

\subsection{Financial market}
Assume that the financial market considered here consists of three assets: one savings account, one stock and one zero-coupon bond with maturity $T_1>T$.

Let $B(t)$ denote the price process of the savings account and assume that  the evolution of $B(t)$ is determined by:
\begin{eqnarray}
dB(t)= r(t)B(t)dt
\end{eqnarray}
with $B(0)=1$. Here, $r(t)$ is short-term interest rate described by
\begin{eqnarray}\label{eq:rt}
dr(t)=a(t)dt+bdW_1(t), t\in [0, T_1], \quad r(0)=r_0
\end{eqnarray}
where $b$ is a positive constant and  $W_1(t)$ is a one dimensional standard Brownian motion. As explicit examples introduced in \citet{korn2002stochastic}, we consider the Ho-Lee model and the Vasicek model, where  $a(t)$ is respectively given by $a(t)=\widetilde a(t)+ b\xi(t)$ and  $ a(t)=\theta(t)-\hat br(t)+ b \xi(t)$. The risk premium $\xi(t)$ is assumed to be deterministic and continuous.

Let $P(t,T_1)$ denote price process of the zero-coupon bond with maturity $ T_1>T $. Then from \citet{korn2002stochastic}, $P(t,T_1)$ satisfies the following SDE,
 \begin{eqnarray}
   dP(t,T_1)=P(t,T_1)\bigg\{\big[r(t)+\xi(t)\sigma_1(t)\big]dt+\sigma_1(t)dW_1(t)\bigg\}, \quad P(0,T_1)=P_0>0,
\end{eqnarray}
where $r(t)$ is given by \eqref{eq:rt} and $\xi(t)$, $\sigma_1(t)$ are deterministic and continuous functions. In what follows,  we shall write $P(t)$ instead of $P(t,T_1)$ for ease of notation. As shown in \citet{korn2002stochastic}, the volatilities of the zero-coupon bond for the Ho-Lee model and Vasicek model are given by $\sigma(t)=-b(T_1-t)$ and $\sigma(t)= \frac {b} {\hat b} [exp\{-\hat b(T_1-t)\}-1]$ respectively.

In addition, we assume that the dynamics of the stock price is described by
 \begin{eqnarray*}
dS(t)=S(t)\big[\mu(t)dt+\sigma_2(t)dW_2(t)\big], \quad S(0)=S_0,
\end{eqnarray*} where ${W_2(t)}$ is one dimension Brownian motion and $\sigma_2(t)$ is deterministic and continuous. Similar to \citet{korn2002stochastic}, we also split up the drift $\mu(t)$ of the stock into a liquidity premium (LP) and a  risk premium (RP) :
\[\mu(t)=\underbrace{r(t)}_{LP}+\underbrace{\mu(t)-r(t)}_{RP}.\]
Let $\lambda(t)$ denote the risk premium of the stock, i.e. $\lambda(t):=\mu(t)-r(t)$. Thus the price process of the stock can be rewritten as
 \begin{eqnarray}
dS(t)=S(t)\bigg\{\big[r(t)+\lambda(t)\big]dt+\sigma_2(t)dW_2(t)\bigg\}.
\end{eqnarray}

\subsection{Surplus process  }
Let $\tilde{R}(t)$ be the real value of the surplus of the insurance company (the value that exclude the impact of price index and inflation). For simplicity, and without loss of generality, we assume that the dynamics of $\tilde{R}(t)$ is described by the following diffusion approximation (DA) model(see \citet{grandell1991aspects},\citet{zeng2011optimal}. 
\begin{eqnarray}
  d\tilde{R}(t)=c(t)dt+\sigma_3(t)dW_3(t),  \quad \quad \quad \quad \tilde{R}(0)=\tilde{R}_0,
\end{eqnarray}
where $\{W_3(t)\}$ is a one dimensional standard Brownian motion, $c(t)>0$ and $\sigma_3(t)$ are deterministic and continuous.  In reality, $c(t)$ and $\sigma_3(t)$ can be regarded as the real premium rate of the insurer and the risk of the insurer respectively. 

Let  $u(t)\in[0,\infty)$ be the proportional reinsurance retention level adopted by the insurance company and assume that the safety loading of the insurance company and the reinsurance company is the same. Thus the surplus process of insurance company after adopting the proportional reinsurance is  
\begin{eqnarray}
d\widetilde R(t)= u(t)c(t)dt+u(t)\sigma_3(t)dW_3(t),  \quad \quad \quad \quad \widetilde R(0)= \widetilde R_0.
\end{eqnarray}

Contrast to other insurance or reinsurance literatures in which they exclude the impact of price index and inflation, we consider the change of surplus according to the price index. Considering long term case, it is natural as well as proper to take the price index into account, since in a long run insurance company need to adjust the amount of claim and premium according to the change of economic condition. Thus we suppose the insurer's surplus process is given by :
\begin{eqnarray}
\nonumber dR(t)&=&\Pi(t)d\widetilde R(t)\\
&=& \Pi(t)u(t)c(t)dt+\Pi(t)u(t)\sigma_3(t)dW_3(t)\ \quad
\end{eqnarray}
with initial data  $R(0)=\Pi(0)\widetilde R(0)$, where $\Pi(t)$ is price index defined in the (2.1).
\\

Finally, we assume the stochastic interest rate, bond price, the price index and the expected inflation rate could be correlated with each other, that is, without loss of generality, we suppose $Cov(W_1(t),W_0(t))=\rho t$,  $\rho\in[-1,1] $. But we assume ${W_2(t)}$, ${W_3(t)}$ are independent Brownian motion with each other and they are independent of ${W_1(t)}$, ${W_0(t)}$. Actually this assumption is proper, since the volatility of claim is not affected by interest rate or price index. For stock, if we suppose the interest rate or price index is correlated with the stock , there will be an additional mixed partial derivative term in HJB equation, but it will not affect the method used in the remainder of the paper.
\subsection{Wealth process}
\quad\quad During the time horizon $[0,T]$ , $T<T_1$, the insurer is allowed to continuously purchase proportional reinsurance and invests all of his (or her) wealth in the financial market. Let the progressively measurable function $\pi_1(t)$, $\pi_2(t)$ are the proportion of the total wealth invested in the bond and stock respectively. Accordingly, $1-\pi_1(t)- \pi_2(t)$ is the proportion of the total wealth invested in the saving account. If we denote by $ \widetilde X(t)$ the wealth of the investor at time $t$ with $ \widetilde X(0) = \widetilde X_0$, then we have
\begin{eqnarray*}
d\widetilde X(t)&=& \pi_1(t) \widetilde X(t) \frac {dP(t)} {P(t)} + \pi_2(t) \widetilde X(t) \frac {dS(t)} {S(t)} + (1-\pi_1(t)- \pi_2(t))\widetilde X(t)\frac {dB(t)} {B(t)}\\
&&+\Pi(t)u(t)c(t)dt+\Pi(t)u(t)\sigma_3(t)dW_3(t)\\
 &=&\widetilde X(t)\bigg\{[r(t)+\pi_1(t)\xi(t)\sigma_1(t)+\pi_2(t)\lambda(t)]dt+ \pi_1(t)\sigma_1(t)dW_1(t) + \pi_2(t)\sigma_2(t)dW_2(t)\bigg\} \\
&&+\Pi(t)u(t)c(t)dt+\Pi(t)u(t)\sigma_3(t)dW_3(t).
\end{eqnarray*}

In the previous section, the nominal price of the real consumption good in the economy at time $t$ is denoted by $\Pi(t)$. The real price of an asset in the economy over a long time is determined by deflating by the price index $\Pi(t)$. The real wealth process including the impact of the inflation is given by $X(t)=\frac{\widetilde X(t) }{\Pi(t)}$. Then by using Ito formula (see \citet{karatzas1991brownian}), we know $X(t)$ follows :
\begin{eqnarray*}
d X(t)&=&\bigg\{X(t)\big[r(t)+\sigma^2_0(t)-I(t)+(\xi(t)\sigma_1(t)-\rho \sigma_1(t) \sigma_0(t))\pi_1(t)+\lambda(t)\pi_2(t)\big]\\
&&+u(t)c(t)\bigg\}dt+u(t)\sigma_3(t)dW_3(t) + X(t)\pi_1(t)\sigma_1(t)dW_1(t)\\
&&+ X(t)\pi_2(t)\sigma_2(t)dW_2(t) - X(t)\sigma_0(t)dW_0(t)
\end{eqnarray*}
with the initial value $X(0)=\frac{\widetilde X(0)}{\Pi(0)}= X_0 $.

For simplicity, let $\eta(t)=\xi(t)-\rho\sigma_0(t)$ and still use symbol $\lambda(t)$ to denote the term $\frac {\lambda(t)}{\sigma_2(t)}$.
 Then we have the final form of the wealth process :
\begin{eqnarray}
\nonumber dX(t)&=&\bigg\{X(t)\big[r(t)+\sigma^2_0(t)-I(t)+\sigma_1(t)\eta(t)\pi_1(t)+\lambda(t)\sigma_2(t)\pi_2(t)\big]\\
&&+u(t)c(t)\bigg\}dt+u(t)\sigma_3(t)dW_3(t) + X(t)\pi_1(t)\sigma_1(t)dW_1(t) \\
\nonumber &&+ X(t)\pi_2(t)\sigma_2(t)dW_2(t) - X(t)\sigma_0(t)dW_0(t)\\
\nonumber \end{eqnarray}
with $X(0)=X_0$.
\\

Finally the admissible control set is given in the definition below.

\begin{defn}

A strategy $ \widetilde\pi(t)=(\pi_1(t), \pi_2(t), u(t)) $ is said to be admissible if $\pi_1(t), \pi_2(t), u(t)$ are progressively measurable processes, and $ \pi_1(t), \pi_2(t) $ are bounded, and $u(t)\geq0$. Further we denote the set of all admissible strategies by $\Theta$.

\end{defn}

\section{Maximizing the expected power utility}
In the setting above, the insurer is interested in choosing a portfolio process to maximize the utility function for his (or  her) terminal wealth .We assume the insurer's preference can be described by a non-log hyperbolic absolute risk aversion (HARA) utility function $U(x)=\frac{1}{p}x^p , 0<p<1,  x>0 $. Obviously, we have $U'>0$ and $U''\leq0$. Now we can formulate the optimization problem :
\[V(t,x,r,I)=\sup_{\widetilde\pi\in\Theta} E \bigg\{ \frac{1}{p}(X^{\widetilde\pi}_T)^p|X_t=x,\quad r_t=r,\quad I_t=I \bigg\},\]
where ${X^{\widetilde\pi}_t}$ is the wealth process under strategy $\widetilde\pi$, and the corresponding state process followed :
\begin{eqnarray}
  \left\{\begin{array}{l}
dX(t)=\bigg\{X(t)\big[r(t)+\sigma^2_0(t)-I(t)+\sigma_1(t)\eta(t)\pi_1(t)+\lambda(t)\sigma_2(t)\pi_2(t)\big]\\
\quad\quad\quad\quad+u(t)c(t)\bigg\}dt+u(t)\sigma_3(t)dW_3(t) + X(t)\pi_1(t)\sigma_1(t)dW_1(t)\\
\quad\quad\quad\quad+ X(t)\pi_2(t)\sigma_2(t)dW_2(t) - X(t)\sigma_0(t)dW_0(t), \quad\quad\quad\quad X(0)=X_0,\\
\\
dr(t)=a(t)dt+bdW_1(t), \quad\quad\quad\quad r(0)=r_0,\\
\\
dI(t)=\beta(t)\big[\alpha(t)-I(t)\big]dt + \bar \sigma_0(t)dW_0(t), \quad\quad\quad\quad I(0)=I_0.
\end{array}\right.
\end{eqnarray}
We define operate:
\begin{eqnarray}
\nonumber \mathcal{A}\psi&=& \nonumber\psi_t+\psi_x\bigg\{x\big[r+\eta\sigma_1\pi_1+\lambda\sigma_2\pi_2-I+\sigma^2_0\big]+uc\bigg\}+\frac{1}{2}\psi_{xx}\big[u^2\sigma^2_3+x^2\pi^2_1\sigma^2_1\\
&&+x^2\pi_2^2\sigma_2^2+x^2\sigma^2_0-2x^2\sigma_1\sigma_0\pi_1\rho\big]+a\psi_r+\frac{1}{2}b^2\psi_{rr}+\psi_I\beta(\alpha-I)
+\frac{1}{2}\bar\sigma^2_0\psi_{II}\\
&&+\big[x\pi_1\sigma_1\bar\sigma_0\rho-x\sigma_0\bar\sigma_0\big]\psi_{xI}+
\nonumber\psi_{xr}\big[x\pi_1\sigma_1b-\rho x\sigma_0b\big]+\psi_{Ir}\rho\bar\sigma_0b.
\end{eqnarray}
Hence the following Hamilton-Jacobi-Bellman equation (HJB) has to be solved :
\begin{eqnarray}
0=\sup_{\widetilde\pi\in\Theta}\bigg\{ \mathcal{A} V(t,x,r,I)\bigg\}.
\end{eqnarray}
Specifically, we can write it explicitly :
\begin{eqnarray}
\nonumber0&=&V_t+V_x\big[r-I+\sigma^2_0(t)\big]x+\frac{1}{2}V_{xx}\sigma^2_0(t)x^2+a(t)V_r+\frac{1}{2}b^2V_{rr}+\beta(t)\big[\alpha(t)-I\big]V_I\\
\nonumber &&+\frac{1}{2}\bar\sigma_0^2(t)V_{II}-\rho\sigma_0(t)bxV_{xr}-\bar\sigma_0(t)\sigma_0(t)xV_{xI}+b\bar\sigma_0(t)\rho V_{rI}+\sup_{|\pi|_1<\delta} \bigg\{\frac{1}{2}V_{xx}\sigma^2_1(t)x^2\pi_1^2\\
&&+\big[V_x\sigma_1(t)\eta(t)x-V_{xx}\rho\sigma_1(t)\sigma_0(t)x^2+V_{xr}b\sigma_1(t)x
+V_{xI}\bar\sigma_0(t)\sigma_1(t)\rho x\big]\pi_1 \bigg\}\\
\nonumber&&+\sup_{|\pi|_2<\delta}\bigg\{\frac{1}{2}V_{xx}\sigma_2^2(t)x^2\pi_2^2+V_x\lambda(t)\sigma_2(t)x\pi_2\bigg\}+\sup_{u\in[0,\infty)}\bigg\{\frac{1}{2}V_{xx}\sigma^2_3(t)u^2+c(t)V_xu\bigg\},\\
\nonumber\\
&&V(T,x,I,r)=\frac{1}{p}x^p,
\end{eqnarray}
where $\delta>0$ will be specified later.

Let us assume that the HJB equation $(3.3)$ has a classical solution G, which satisfies condition $G_x>0$ and $G_{xx}<0$.

We get the following candidate for the optimal bond position:
\begin{eqnarray}
\pi_1^*(t) = -\frac{\eta(t)}{\sigma_1(t)}\frac{V_x}{xV_{xx}}-\frac{\bar\sigma_0(t)\rho}{\sigma_1(t)}\frac{V_{Ix}}{xV_{xx}}-\frac{b}{\sigma_1(t)}\frac{V_{xr}}{xV_{xx}}+\frac{\sigma_0(t)\rho}{\sigma_1(t)},
\end{eqnarray}
the optimal stock position:
\begin{eqnarray}
\pi_2^*(t)=-\frac{V_x}{xV_{xx}}\frac{\lambda(t)}{\sigma_2(t)},
\end{eqnarray}
and the candidate for the optimal reinsurance proportion
\begin{eqnarray}
u^*(t)=-\frac{V_x}{V_{xx}}\frac{c(t)}{\sigma_3^2(t)}.
\end{eqnarray}

From the form of the equation above, we conjecture that the solution $G$ has the form :
\begin{eqnarray}
G(t,x,r,I)=g(t,r,I)\frac{x^p}{p}
\end{eqnarray}
with $g(T,r,I)=1$ for all $I$ and $r$. After simple calculation we obtain these results :
\begin{eqnarray*}
&&G_t=g_t\frac{x^p}{p}, \quad \quad G_x=gx^{p-1}, \quad \quad G_{xx}=g(p-1)x^{p-2}, \quad \quad G_{r}=g_{r}\frac{x^p}{p}, \\
&&G_{rr}=g_{rr}\frac{x^p}{p}, \quad \quad G_{I}=g_{I}\frac{x^p}{p}, \quad \quad G_{II}=g_{II}\frac{x^p}{p}, \quad \quad G_{xr}=g_{r}x^{p-1},\\
&&G_{xI}=g_Ix^{p-1}, \quad \quad G_{rI}=g_{rI}\frac{x^p}{p}.
\end{eqnarray*}
Plug them into (3.4) lead to another equation for $g$ of the form:
\begin{eqnarray}
\nonumber 0&=&\frac{g_t}{p}+g(r+\sigma^2_0-I)+\frac{1}{2}\sigma_0^2g(p-1)+ag_r\frac{1}{p}+\frac{1}{2}b^2\frac{g_{rr}}{p}+\frac{g_I}{p}\beta(\alpha-I)\\
\nonumber &&+\frac{1}{2}\bar\sigma^2_0\frac{g_{II}}{p}-\sigma_0\bar\sigma_0g_I-g_r\rho\sigma_0b+\frac{g_{Ir}}{p}b\rho\bar\sigma_0-\frac{g}{p-1}\frac{\eta^2}{2}-\frac{1}{2}g(p-1)\sigma^2_0\rho^2\\
&&-\frac{g_I^2}{g(p-1)}\frac{\bar\sigma^2_0\rho^2}{2}-\frac{g_r^2}{g(p-1)}\frac{b^2}{2}+g\eta\sigma_0\rho-\frac{g_I}{p-1}\eta\rho\bar\sigma_0-\frac{g_r}{p-1}b\eta\\
\nonumber &&+g_rb\sigma_0\rho+ g_I\bar\sigma_0\sigma_0\rho^2-\frac{g_r g_I}{g(p-1)}\bar\sigma_0\rho b - \frac{g}{p-1}\frac{\lambda^2}{2} - \frac{g}{p-1}\frac{c^2}{2\sigma_3^2}.
\end{eqnarray}

Next we use the ansatz :
\begin{eqnarray}
g(t,r,I)=f(t)e^{k(t)r+z(t)I}
\end{eqnarray}
with terminal value $f(T)=1$,  $k(T)=0$, $z(T)=0$. After simple calculation, we get these results:
\begin{eqnarray*}
&& g_t=f'e^\Delta+fe^\Delta(k'r+z'I), \quad\quad g_r=fke^\Delta, \quad\quad g_I=fz e^\Delta,\\
&& g_{Ir}=fkze^{\Delta}, \quad\quad g_{rr}=fk^2e^\Delta, \quad\quad g_{II}=fz^2e^\Delta,
\end{eqnarray*}
where denote $\Delta:=k(t)r+z(t)I$ for simplicity. Inserting them into (3.10) and simplification yield:
 \begin{eqnarray}
\nonumber 0&=&[\frac{k'}{p}+1]fr+[\frac{z'}{p}-\frac{\beta z}{p}-1]fI + \frac{f'}{p}+\bigg\{\sigma_0^2+\frac{1}{2}(p-1)\sigma_0^2+\frac{b^2k^2}{2p}\\
\nonumber\\
\nonumber&&+\frac{\alpha\beta}{p}z+\frac{\bar\sigma_0^2z^2}{2p}-\bar\sigma_0\sigma_0z-k\rho\sigma_0 b+\frac{b\rho\bar\sigma_0}{p}kz-\frac{\eta^2}{2(p-1)}-\frac{p-1}{2}\rho^2\sigma_0^2\\
\\
\nonumber&&-\frac{\rho^2\bar\sigma^2_0}{2(p-1)}z^2 - \frac{b^2}{2(p-1)}k^2+\sigma_0\rho\eta-\frac{\eta\rho\bar\sigma_0}{p-1}z-\frac{\eta b}{p-1}k+\bar\sigma_0\sigma_0 \rho^2z\\
\nonumber\\
\nonumber&&+\sigma_0b\rho k-\frac{\bar\sigma_0\rho b}{p-1}kz - \frac{\lambda^2}{2(p-1)}-\frac{c^2}{2\sigma_3^2(p-1)}\bigg \}f+\frac{ak}{p}f.
\end{eqnarray}
\\

In the Ho-lee model the drift $a(t)$ of the stochastic interest rate is deterministic and continuous, but because of containing the term $r(t)$, the drift $a(t)$ in the vasicek model is stochastic. We must treat them separately.\\

\emph{\textbf{Ho-lee model}}:\\

In Ho-lee model E.q.(3.12) has the form

  \begin{eqnarray}
\nonumber 0&=&[\frac{k'}{p}+1]fr+[\frac{z'}{p}-\frac{\beta z}{p}-1]fI + \frac{f'}{p}+\bigg\{\sigma_0^2+\frac{1}{2}(p-1)\sigma_0^2+\frac{ak}{p}+\frac{b^2k^2}{2p}\\
\nonumber\\
\nonumber&&+\frac{\alpha\beta}{p}z+\frac{\bar\sigma_0^2z^2}{2p}-\bar\sigma_0\sigma_0z-k\rho\sigma_0 b+\frac{b\rho\bar\sigma_0}{p}kz-\frac{\eta^2}{2(p-1)}-\frac{p-1}{2}\rho^2\sigma_0^2\\
\\
\nonumber&&-\frac{\rho^2\bar\sigma^2_0}{2(p-1)}z^2 - \frac{b^2}{2(p-1)}k^2+\sigma_0\rho\eta-\frac{\eta\rho\bar\sigma_0}{p-1}z-\frac{\eta b}{p-1}k+\bar\sigma_0\sigma_0 \rho^2z\\
\nonumber\\
\nonumber&&+\sigma_0b\rho k-\frac{\bar\sigma_0\rho b}{p-1}kz - \frac{\lambda^2}{2(p-1)}-\frac{c^2}{2\sigma_3^2(p-1)}\bigg \}f.
\end{eqnarray}

We define $h(t):=\bigg\{....\bigg\}$ in the E.q.(3.13). Thus we have to solve three ordinary differential equation(ODE):
\begin{eqnarray}
  \left\{\begin{array}{l}
k'(t)=-p,\\
\\
k(T)=0,
\end{array}\right.
\end{eqnarray}
lead to $k(t)=p(T-t)$.
\begin{eqnarray}
  \left\{\begin{array}{l}
z'(t)-\beta(t)z(t)-p=0,\\
\\
z(T)=0,
\end{array}\right.
\end{eqnarray}
\\
which lead to $ z(t)=-pe^{\int_0^t\beta(s)ds}\int_t^Te^{-\int_0^s\beta(v)dv}ds$.
\\
\begin{eqnarray}
  \left\{\begin{array}{l}
f'(t)+ph(t)f(t)=0,\\
\\
f(T)=1,
\end{array}\right.
\end{eqnarray}
lead to $f(t)=e^{-p[H(t)-H(T)]}$ , where $H(t)$ is a primitive of $h(t)$.

Combining (3.9),(3.11) with the solution of (3.16) , we have a candidate for the optimal value :
\begin{eqnarray}
G(t,x,r,I)=\frac{1}{p}\exp\bigg\{-p\big[H(t)-H(T)\big]\bigg\}exp\bigg\{k(t)r+z(t)I\bigg\}x^p,
\end{eqnarray}
where $k(t)$ and $z(t)$ is the solution of E.q.(3.14) and E.q.(3.15) respectively, and the corresponding optimal policies :
\begin{eqnarray}
  \left\{\begin{array}{l}
u^*(t)=-\frac{c(t)}{\sigma^2_3(t)}\frac{1}{p-1}x,\\
\\
\pi^*_1(t)=-\frac{\eta(t)}{\sigma_1(t)}\frac{1}{p-1}-\frac{b}{\sigma_1(t)}\frac{p}{p-1}(T-t)+\frac{\rho\sigma_0(t)}{\sigma_1(t)}\\
\quad\quad\quad\quad-\frac{\rho\bar\sigma_0(t)}{\sigma_1(t)}\frac{p}{p-1}e^{\int_0^t\beta(s)ds}\int_t^Te^{-\int_0^s\beta(v)dv}ds,\\
\\
\pi_2^*(t)=-\frac{\lambda(t)}{\sigma_2(t)}\frac{1}{p-1}.
\end{array}\right.
\end{eqnarray}
Note that $\lambda(t), \sigma_1(t), \sigma_2(t), \eta(t) $ are all deterministic and continuous on the time interval $[0,T]$, therefor $\pi^*_1(t)$ and $\pi^*_2(t)$ are deterministic and continuous on $[0,T]$, thus obviously they are bounded, and explicitly $\sigma_1(t)=-b(T_1-t) $. Further, because $c(t)>0$ and $\sigma_3^2(t)>0$, we have $u^*(t)>0$, which satisfies the constraint. On the other hand, because of $0<p<1$, we have $V_{xx}<0$, which satisfies the hypothesis ahead.\\

\emph{\textbf{Vasicek model}:}\\

In the Vasicek model, the drift of the stochastic interest rate has the form of $a(t)=\theta(t)-\hat br(t)+b\xi(t)$, then the E.q.(3.12) is equivalent to

  \begin{eqnarray}
\nonumber 0&=&[\frac{k'}{p}-\frac{\hat b}{p}k+1] fr+[\frac{z'}{p}-\frac{\beta z}{p}-1] f I + \frac{ f'}{p}+\bigg\{\sigma_0^2+\frac{1}{2}(p-1)\sigma_0^2+\frac{(\theta+b\xi)k}{p}+\frac{b^2k^2}{2p}\\
\nonumber\\
\nonumber&&+\frac{\alpha\beta}{p}z+\frac{\bar\sigma_0^2z^2}{2p}-\bar\sigma_0\sigma_0z-k\rho\sigma_0 b+\frac{b\rho\bar\sigma_0}{p} kz-\frac{\eta^2}{2(p-1)}-\frac{p-1}{2}\rho^2\sigma_0^2\\
\\
\nonumber&&-\frac{\rho^2\bar\sigma^2_0}{2(p-1)}z^2 - \frac{b^2}{2(p-1)}k^2+\sigma_0\rho\eta-\frac{\eta\rho\bar\sigma_0}{p-1}z-\frac{\eta b}{p-1}k+\bar\sigma_0\sigma_0 \rho^2z\\
\nonumber\\
\nonumber&&+\sigma_0b\rho k-\frac{\bar\sigma_0\rho b}{p-1} kz - \frac{\lambda^2}{2(p-1)}-\frac{c^2}{2\sigma_3^2(p-1)}\bigg \} f.
\end{eqnarray}

We also define $ h(t):=\bigg\{....\bigg\}$ in the E.q.(3.19). Thus we have to solve three ordinary differential equation(ODE):
\begin{eqnarray}
  \left\{\begin{array}{l}
k'(t)- \hat b k(t)+p=0,\\
\\
k(T)=0,
\end{array}\right.
\end{eqnarray}
lead to $k(t)=\frac{p}{\hat b}[1-\exp\{\hat b(t-T)\}]$.
\begin{eqnarray}
  \left\{\begin{array}{l}
z'(t)-\beta(t)z(t)-p=0,\\
\\
z(T)=0,
\end{array}\right.
\end{eqnarray}
\\
which lead to $z(t)=-pe^{\int_0^t\beta(s)ds}\int_t^Te^{-\int_0^s\beta(v)dv}ds$.
\\
\begin{eqnarray}
  \left\{\begin{array}{l}
 f'(t)+ph(t)f(t)=0,\\
\\
f(T)=1,
\end{array}\right.
\end{eqnarray}
lead to $f(t)=e^{-p[H(t)-H(T)]}$ , where $H(t)$ is a primitive of $h(t)$.

Combining (3.9),(3.11) with the solution of (3.22) , we have a candidate for the optimal value :
\begin{eqnarray}
G(t,x,r,I)=\frac{1}{p}exp\bigg\{-p\big[H(t)-H(T)\big]\bigg\}\exp\bigg\{ k(t)r+z(t)I\bigg\}x^p,
\end{eqnarray}
where $k(t)$ and $I(t)$ is the solution of E.q.(3.20) and E.q.(3.21) respectively, and the corresponding optimal policies :
\begin{eqnarray}
  \left\{\begin{array}{l}
u^*(t)=-\frac{c(t)}{\sigma^2_3(t)}\frac{1}{p-1}x,\\
\\
\pi^*_1(t)=-\frac{\eta(t)}{\sigma_1(t)}\frac{1}{p-1}-\frac{1}{\sigma_1(t)}\frac{p}{p-1}[1-exp\{ b(t-T)\}]\\
\quad\quad\quad-\frac{\rho\bar\sigma_0(t)}{\sigma_1(t)}\frac{p}{p-1}e^{\int_0^t\beta(s)ds}\int_t^Te^{-\int_0^s\beta(v)dv}ds+\frac{\rho\sigma_0(t)}{\sigma_1(t)},\\
\\
\pi_2^*(t)=-\frac{\lambda(t)}{\sigma_2(t)}\frac{1}{p-1}.
\end{array}\right.
\end{eqnarray}
The same as Ho-lee model, constraints $u^*(t)>0$, $V_{xx}<0$ are satisfied, and $\lambda(t), \sigma_1(t), \sigma_2(t), \eta(t) $ are all deterministic and continuous on the time interval $[0,T]$, therefor $\pi^*_1(t)$ and $\pi^*_2(t)$ are deterministic and continuous on $[0,T]$, thus obviously they are bounded. Note in Vasicek model we have $\sigma_1(t)= \frac {b} {\hat b} [\exp\{-\hat b(T_1-t)\}-1]$ .

\section{ Verification theorem}
Due to the presence of the production $rx$ and $Ix$ in the state process of $(2.10)$, the usual verification theorem which require Lipschitz condition and linear growth condition are not applicable to our situation, since the wealth process, the stochastic interest rate process and the inflation process have the possibility of unbound.  There are some methods can be used to deal with this sort of problem. In the \citet{korn2002stochastic}, a suitable verification theorem is given to overcome this difficult. It shows that under linear controlled SDEs, the assumption required in the standard verification can be weaken. In the following, however, we adopt another way of proving the candidates given by (3.17) and (3.23) are optimal for the optimization problem. This method is similar to the one used in \citet{li2009optimal,kraft2009optimal}, and comparing with the method in \citet{korn2002stochastic}, which is more delicate.

First we have the wealth process under optimal strategy $\widetilde\pi^*(t)=(\pi_1^*(t),\pi_2^*(t),u^*(t)),$
\begin{eqnarray}
\nonumber dX^*_t&=&X^*_t\bigg\{\big[r(t)+\pi_1^*(t)\eta(t)\sigma_1(t)+\pi_2^*(t)\sigma_2(t)\lambda(t)-I(t)+\sigma_0^2(t)\big]dt\\
\nonumber &&+\sigma_1(t)\pi_1^*(t)dW_1(t)+\sigma_2(t)\pi_2^*(t)dW_2(t)-\sigma_0(t)dW_0(t)\bigg\}+u^*(t)c(t)dt + u^*(t)\sigma_3(t)dW_3(t)\\
&=&X^*_t\bigg\{\big[r(t)+\pi_1^*(t)\eta(t)\sigma_1(t)+\pi_2^*(t)\sigma_2(t)\lambda(t)-I(t)+\sigma_0^2(t)+\frac{c^2(t)}{\sigma_3^2(t)}\frac{1}{1-p}\big]dt\\
\nonumber &&+\frac{c(t)}{\sigma_3(t)}\frac{1}{1-p}dW_3(t)+\sigma_1(t)\pi^*_1(t)dW_1(t) +\sigma_2(t)\pi^*_2(t)dW_2(t)-\sigma_0(t)dW_0(t)\bigg\},
\end{eqnarray}
where we have used the result $u^*(t)=-\frac{c(t)}{\sigma^2_3(t)}\frac{1}{p-1}X_t^*$ in $(3.18)$ and (3.24).

We can solve the SDE with initial value $X(0)=X_0$ to get the following form:
\begin{eqnarray}
\nonumber X^*_t&=&X_0\exp\bigg\{\int_0^t\big[r(s)+\pi_1^*(s)\eta(s)\sigma_1(s)+\pi_2^*(s)\sigma_2(s)\lambda(s)-I(s)+\sigma_0^2(s)\\
\nonumber &&+\frac{c^2(s)}{(1-p)\sigma_3^2(s)}-\frac{c^2(s)}{2(1-p)^2\sigma_3^2(s)}-\frac{\sigma_1^2(s)(\pi_1^*(s))^2}{2}-\frac{\sigma_2^2(s)(\pi_2^*(s))^2}{2}-\frac{\sigma_0^2(s)}{2}\\
\nonumber &&+\sigma_1(s)\sigma_0(s)\pi^*_1(s)\rho\big]ds+\frac{1}{1-p}\int_0^t\frac{c(s)}{\sigma_3(s)}dW_3(s)+\int_0^t\sigma_1(s)\pi_1^*(s)dW_1(s)\\
\nonumber &&+\int_0^t\sigma_2(s)\pi_2^*(s)dW_2(s)-\int_0^t\sigma_0(s)dW_0(s)\bigg\}\\
\nonumber &=&D_1(t)\cdot \exp\bigg\{\int_0^tr(s)ds-\int_0^tI(s)ds+\frac{1}{1-p}\int_0^t\frac{c(s)}{\sigma_3(s)}dW_3(s)+\int_0^t\sigma_1(s)\pi_1^*(s)dW_1(s)\\
&&+\int_0^t\sigma_2(s)\pi_2^*(s)dW_2(s)-\int_0^t\sigma_0(s)dW_0(s)\bigg\},
\end{eqnarray}
where we let
\begin{eqnarray}
\nonumber D_1(t)&=&X_0\exp\bigg\{\int_0^t\big[\pi_1^*(s)\eta(s)\sigma_1(s)+\pi_2^*(s)\sigma_2(s)\lambda(s)+\sigma_0^2(s)\\
\nonumber&&+\frac{c^2(s)}{(1-p)\sigma_3^2(s)}
-\frac{c^2(s)}{2(1-p)^2\sigma_3^2(s)}-\frac{\sigma_1^2(s)(\pi_1^*(s))^2}{2}-\frac{\sigma_2^2(s)(\pi_1^*(s))^2}{2}-\frac{\sigma_0^2(s)}{2}\\
&&+\sigma_1(s)\sigma_0(s)\pi^*_1(s)\rho\big]ds\bigg\}.
\end{eqnarray}
Note that $D_1(t)$ is  deterministic and continuous on the interval $[0,T]$, since all parameters including $\pi_1^*(t),\pi_2^*(t)$ in the integral are deterministic and continuous.

We have the candidate value function under optimal policy :
\begin{eqnarray}
\nonumber G(t,X^*_t,r_t,I_t)&=&g(t,r,I)\frac{1}{p}(X^*_t)^p\\
\nonumber&=&\frac{1}{p}\exp\bigg\{-p\big[H(t)-H(T)\big]\bigg\}\exp\bigg\{k(t)r(t)+z(t)I(t)\bigg\}(X_t^*)^p\\
\nonumber&=&\frac{1}{p}\exp\bigg\{-p\big[H(t)-H(T)\big]\bigg\}D_1^p(t)\exp\bigg\{k(t)r(t)+z(t)I(t)\bigg\}\\
\nonumber&&\cdot \exp\bigg\{p\int_0^tr(s)ds-p\int_0^tI(s)ds+\frac{p}{1-p}\int_0^t\frac{c(s)}{\sigma_3(s)}dW_3(s)+p\int_0^t\sigma_1(s)\pi_1^*(s)dW_1(s)\\
\nonumber&&+p\int_0^t\sigma_2(s)\pi^*_2(s)dW_2(s)-p\int_0^t\sigma_0(s)dW_0(s) \bigg\}\\
\nonumber&=&D_2(t)\exp\bigg\{k(t)r(t)+p\int_0^tr(s)ds+p\int_0^t\sigma_1(s)\pi_1^*(s)dW_1(s)\bigg\}\\
\nonumber&&\cdot \exp\bigg\{z(t)I(t)-p\int_0^tI(s)ds-p\int_0^t\sigma_0(s)dW_0(s)\bigg\}
 \exp\bigg\{\frac{p}{1-p}\int_0^t\frac{c(s)}{\sigma_3(s)}dW_3(s)\bigg\}\\
&&\cdot \exp\bigg\{p\int_0^t\sigma_2(s)\pi_2^*(s)dW_2(s)\bigg\},
\end{eqnarray}
where $D_2(t)=\frac{1}{p}\exp\bigg\{-p\big[H(t)-H(T)\big]\bigg\}D_1^p(t)$ and obviously it is deterministic and continuous.
\\
\\
Remark:

Strictly speaking, Ho-lee model and Vasicek model have different expression in $k(t)$, $f(t)$, optimal strategy $\pi_2^*(t)$ and volatility $\sigma_1(t)$. As we have seen, however, they have the same form of candidate value function $G$ and the same form of wealth process $X(t)$ such that they will have the same derivation, thus in the derivation above, we didn't distinguish the two models in terms of the usage of notation.
\\

 Now we give two lemma which are very useful in the verification theorem following. The first one aim at Ho-lee model, the other is about vasicek model.
\begin{lem}\label{sec:lem1}
In the setting of Ho-lee model, we assume that $G$ and $\widetilde\pi^*(t)=(\pi_1^*(t),\pi_2^*(t),u^*(t))$ are given by (3.17) and (3.18) respectively. Then the sequence $\{ G(\tau_n,X_{\tau_n}^{\widetilde\pi^*} , I_{\tau_n} , r_{\tau_n})\}_{n\in N}$ is uniformly integrable for all sequence of stopping times $\{\tau_n\}_{n\in N}$.
\end{lem}
\begin{proof}
 We let $X^*_t:=X^{\widetilde\pi^*}_t $, according to the form given by (4.4). For every fixed $q>1$ we have
 \begin{eqnarray}
 \nonumber|G(t,X^*_t,r_t,I_t)|^q&=&|D_2(t)|^q\exp\bigg\{qk(t)r(t)+qp\int_0^tr(s)ds+qp\int_0^t\sigma_1(s)\pi_1^*(s)dW_1(s)\bigg\}\\
\nonumber&&\exp\bigg\{qz(t)I(t)-qp\int_0^tI(s)ds-qp\int_0^t\sigma_0(s)dW_0(s)\bigg\}\\
&&\exp\bigg\{\frac{qp}{1-p}\int_0^t\frac{c(s)}{\sigma_3(s)}dW_3(s)\bigg\}\exp\bigg\{qp\int_0^t\sigma_2(s)\pi_2^*(s)dW_2(s)\bigg\}.
 \end{eqnarray}

 At first, combining with the result $k(t)=p(T-t)$ we have
 \begin{eqnarray}
\nonumber&&\exp\bigg\{qk(t)r(t)+qp\int_0^tr(s)ds+qp\int_0^t\sigma_1(s)\pi_1^*(s)dW_1(s)\bigg\}\\
\nonumber&&=\exp\bigg\{qp(T-t)r(t)+qp\int_0^tr(s)ds+qp\int_0^t\sigma_1(s)\pi_1^*(s)dW_1(s)\bigg\}\\
&&=\exp\bigg\{qpTr(t)\bigg\}\exp\bigg\{-qptr(t)\bigg\}\exp\bigg\{qp\int_0^tr(s)ds+qp\int_0^t\sigma_1(s)\pi_1^*(s)dW_1(s)\bigg\}.
 \end{eqnarray}
 With the equality
 \[\exp\bigg\{-qptr(t)\bigg\}=\exp\bigg\{-qp\int_0^tsdr(s)-qp\int_0^tr(s)ds\bigg\},\]
 we have
 \[(4.6)=\exp\bigg\{qp\int_0^t\sigma_1(s)\pi^*_1(s)dW_1(s)\bigg\}\exp\bigg\{qpTr(t)-qp\int_0^tsdr(s)\bigg\}.\]
 From the SDE of the stochastic interest rate
 \[dr(t)=a(t)dt+bdW_1(t) \quad\quad with \quad r(0)=r_0,\]
 we finally obtain
 \begin{eqnarray}
 \nonumber&&\exp\bigg\{qk(t)r(t)+qp\int_0^tr(s)ds+qp\int_0^t\sigma_1(s)\pi_1^*(s)dW_1(s)\bigg\}\\
\nonumber&=&\exp\bigg\{qp\int_0^t\sigma_1(s)\pi^*_1(s)dW_1(s)\bigg\}\exp\bigg\{qpTr_0+qpT\int_0^ta(s)ds+qpT\int_0^tbdW_1(s)\\
\nonumber &&-qp\int_0^ts\big[a(s)ds+bdW_1(s)\big]\bigg\}\\
\nonumber &=&\exp\bigg\{qp\int_0^t\sigma_1(s)\pi^*_1(s)dW_1(s)\bigg\}\exp\bigg\{qpTr_0+qp\big[\int_0^ta(s)(T-s)ds+\int_0^tb(T-s)dW^1_s\big]\bigg\}\\
\nonumber &=&\exp\bigg\{qpTr_0\bigg\}\exp\bigg\{qp\int_0^ta(s)(T-s)ds\bigg\}\exp\bigg\{qp\int_0^t\big[\sigma_1(t)\pi_1^*(s)+b(T-s)\big]dW_1(s)\bigg\}.\\
 \end{eqnarray}

Next we consider the term :
 \begin{eqnarray}
\exp\bigg\{qz(t)I(t)-qp\int_0^tI(s)ds-qp\int_0^t\sigma_0(s)dW_0(s)\bigg\}.
 \end{eqnarray}
By Ito formula
 \begin{eqnarray}
\nonumber z(t)I(t)&=&z(0)I_0+\int_0^tI(s)z'(s)ds+\int_0^tz(s)dI(s)\\
\nonumber &=&z(0)I_0+\int_0^tI(s)z'(s)ds+\int_0^tz(s)\bigg\{\beta(s)\big[\alpha(s)-I(s)\big]ds+\bar\sigma_0(s)dW_0(s)\bigg\}\\
\nonumber &=&z(0)I_0+\int_0^tI(s)z'(s)ds-\int_0^tz(s)\beta(s)I(s)ds\\
&&+\int_0^tz(s)\beta(s)\alpha(s)ds+\int_0^t\bar\sigma_0(s)z(s)dW_0(s),
 \end{eqnarray}
 plug (4.9) into (4.8), then
  \begin{eqnarray}
\nonumber (4.8)&=& \exp\bigg\{qz(0)I_0+q\int_0^t[z'(s)-\beta(s)z(s)-p]I(s)ds+q\int_0^t\beta(s)\alpha(s)z(s)ds\\
  &&+q\int_0^t\bar\sigma_0(s)z(s)dW_0(s)-qp\int_0^t\sigma_0(s)dW_0(s)\bigg \}.
 \end{eqnarray}
Combining with equation (3.15), we get
\begin{eqnarray}
\nonumber(4.8)&=&\exp\bigg\{qz(0)I_0+q\int_0^t\beta(s)\alpha(s)z(s)ds+q\int_0^t\bar\sigma_0(s)z(s)dW_0(s)-qp\int_0^t\sigma_0(s)dW_0(s) \bigg\}\\
   &=&\exp\bigg\{qz(0)I_0+q\int_0^t\beta(s)\alpha(s)z(s)ds\bigg\} \exp\bigg\{ \int_0^t q\big[ \bar\sigma_0(s)z(s)-p\sigma_0(s)\big]dW_0(s)\bigg\}.
\end{eqnarray}
Put these results, (4.7) and (4.11), into (4.5). We get:
\begin{eqnarray}
\nonumber(4.5)&=&|D_2(t)|^q\exp\bigg\{qpTr_0+qz(0)I_0+qp\int_0^ta(s)(T-s)ds+q\int_0^t\beta(s)\alpha(s)z(s)ds\bigg\}\\
\nonumber&&\exp\bigg\{\frac{qp}{1-p}\int_0^t\frac{c(s)}{\sigma_3(s)}dW_3(s)+qp\int_0^t\big[\sigma_1(s)\pi_1^*(s)+b(T-s)\big]dW_1(s)\\
\nonumber&&+qp\int_0^t\sigma_2(s)\pi^*_2(s)dW_2(s)+q\int_0^t\big[\bar\sigma_0(s)z(s)-p\sigma_0(s)\big]dW_0(s) \bigg   \}.
\end{eqnarray}
For the sum of the stochastic integral above, we construct martingale
\begin{eqnarray}
\nonumber(4.5)&=&|D_2(t)|^q\exp\bigg\{qpTr_0+qz(0)I_0+qp\int_0^ta(s)(T-s)ds+q\int_0^t\beta(s)\alpha(s)z(s)ds\bigg\}\\
\nonumber&&\exp\bigg\{\frac{1}{2}\frac{q^2p^2}{(1-p)^2}\int_0^t\frac{c^2(s)}{\sigma_3^2(s)}ds+\frac{1}{2}q^2p^2\int_0^t\big[\sigma_1(s)\pi_1^*(s)+b(T-s)\big]^2ds\\
\nonumber&&+\frac{1}{2}q^2p^2\int_0^t\sigma^2_2(s)(\pi^*_2(s))^2ds+\frac{1}{2}q^2\int_0^t\big[\bar\sigma_0(s)z(s)-p\sigma_0(s)\big]^2ds\\
\nonumber&&+q^2p\rho\int_0^t\big[\sigma_1(s)\pi^*_1(s)+b(T-s)\big]\big[\bar\sigma_0(s)z(s)-p\sigma_0(s)\big]ds  \bigg \}\cdot M(t)\\
&=&D_3(t)\cdot M(t),
\end{eqnarray}
where $D_3(t)$ is deterministic and continuous on time interval [0,T], therefore $D_3(t)$ is bounded. On the other hand, $M(t)$ is a martingale, and specifically,
\begin{eqnarray}
\nonumber M(t)&=&\exp\bigg\{-\frac{1}{2}\frac{q^2p^2}{(1-p)^2}\int_0^t\frac{c^2(s)}{\sigma_3^2(s)}ds-\frac{1}{2}q^2p^2\int_0^t\big[\sigma_1(s)\pi_1^*(s)+b(T-s)\big]^2ds\\
\nonumber&&-\frac{1}{2}q^2p^2\int_0^t\sigma^2_2(s)(\pi^*_2(s))^2ds-\frac{1}{2}q^2\int_0^t\big[\bar\sigma_0(s)z(s)-p\sigma_0(s)\big]^2ds\\
\nonumber&&-q^2p\rho\int_0^t\big[\sigma_1(s)\pi^*_1(s)+b(T-s)\big]\big[\bar\sigma_0(s)z(s)-p\sigma_0(s)\big]ds \bigg\}\\
\nonumber&&\exp\bigg\{\frac{qp}{1-p}\int_0^t\frac{c(s)}{\sigma_3(s)}dW_3(s)+qp\int_0^t[\sigma_1(s)\pi_1^*(s)+b(T-s)]dW_1(s)\\
&&+qp\int_0^t\sigma_2(s)\pi^*_2(s)dW_2(s)+q\int_0^t\big[\bar\sigma_0(s)z(s)-p\sigma_0(s)\big]dW_0(s)  \bigg  \}.
\end{eqnarray}
Finally, due to the optional stopping theorem (\citet{karatzas1991brownian}), we obtain that for all stopping times $\tau_n$ with $0\leq\tau_n\leq T$,
\begin{eqnarray*}
E(|G(\tau_n,X_{\tau_n}^{\widetilde\pi^*},r_{\tau_n},I_{\tau_n})|^q)=E\big[D_3(\tau_n)\cdot M(\tau_n)\big]\leq\sup_{t\in[0,T]}D_3(t)\cdot E\big[M(\tau_n)\big]\leq\sup_{t\in[0,T]}D_3(t)<\infty.
\end{eqnarray*}

\end{proof}

\begin{lem}\label{sec:lem1}
In the setting of Vasicek model, we assume that $G$ and $\widetilde\pi^*(t)=(\pi_1^*(t),\pi_2^*(t),u^*(t))$ are given by $(3.23)$ and $(3.24)$ respectively. Then the sequence $\{ G(\tau_n,X_{\tau_n}^{\widetilde\pi^*} , I_{\tau_n} , r_{\tau_n})\}_{n\in N}$ is uniformly integrable for all sequence of stopping times $\{\tau_n\}_{n\in N}$.
\end{lem}

\begin{proof}
The proof here is similar to lemma 4.1, the difference between the two case is the form of the drift of stochastic interest rate.
In Ho-lee model $a(t)$ is deterministic and continuous on interval [0,T], but in Vasicek model the drift is stochastic, explicitly we have $a(t)=\theta(t)-\hat b r(t)+b\xi(t)$. Now, for simplicity, we only give the key step which is different from lemma 4.1.

Let $X^*_t:=X^{\widetilde\pi^*}_t $, according to the form given by (4.4), for every fixed $q>1$ we have
 \begin{eqnarray}
 \nonumber|G(t,X^*_t,r_t,I_t)|^q&=&|D_2(t)|^q\exp\bigg\{qk(t)r(t)+qp\int_0^tr(s)ds+qp\int_0^t\sigma_1(s)\pi_1^*(s)dW_1(s)\bigg\}\\
\nonumber&&\exp\bigg\{qz(t)I(t)-qp\int_0^tI(s)ds-qp\int_0^t\sigma_0(s)dW_0(s)\bigg\}\\
&&\exp\bigg\{\frac{qp}{1-p}\int_0^t\frac{c(s)}{\sigma_3(s)}dW_3(s)\bigg\}\exp\bigg\{qp\int_0^t\sigma_2(s)\pi_2^*(s)dW_2(s)\bigg\}.
 \end{eqnarray}
 Combining with the result in $(3.22)$, i.e. $k(t)=\frac{p}{\hat b}[1-e^{\hat b(t-T)}]$
 \begin{eqnarray}
 \nonumber &&\exp\bigg\{qk(t)r(t)+qp\int_0^t r(s)ds+qp\int_0^t \sigma_1(s)\pi_1^*(s)dW_1(s)\bigg\}\\
 \nonumber &=&\exp\bigg\{\frac{qp}{\hat b}\big[1-e^{\hat b(t-T)}\big]r(t)+qp\int_0^t r(s)ds + qp\int_0^t \sigma_1(s)\pi_1^*(s)dW_1(s)  \bigg  \}\\
 \nonumber &=&\exp\bigg\{\frac{qp}{\hat b}r(t)\bigg\}\exp\bigg\{-\frac{qp}{\hat b}e^{\hat b (t-T)}r(t)\bigg\}\\
  &&\cdot \exp\bigg\{qp\int_0^t r(s)ds+ qp\int_0^t \sigma_1(s)\pi_1^*(s)dW_1(s)\bigg\}.
 \end{eqnarray}
 Note in vasicek model $r(t)$ is described by SDE:
 \[dr(t)=\big[\theta(t)-\hat b r(t)+b\xi(t)\big]dt+bdW_1(t).\]
Then by Ito formula we have
 \begin{eqnarray}
 \nonumber e^{ \hat b(t-T)}r(t)&=& e^{-\hat b T} r_0 + \int_0^t \hat b e^{\hat b(s-T)}r(s) ds + \int_0^t e^{\hat b (s-T)}d r(s)\\
 \nonumber &=&e^{-\hat bT}r_0 + \int_0^t \hat b e^{\hat b(s-T)}r(s) ds + \int_0^t e^{\hat b(s-T)}[\theta(s)-\hat br(s) + b\xi(s)]ds \\
 \nonumber &&+ \int_0^t b e^{\hat b(s-T)}dW_1(s)\\
 &=& e^{-\hat bT}r_0 + \int_0^t e^{\hat b(s-T)}[\theta(s)+ b\xi(s)]ds + \int_0^t b e^{\hat b(s-T)}dW_1(s),
\end{eqnarray}
plug $(4.16)$ into $(4.15)$
\begin{eqnarray*}
E.q.(4.15) &=& \exp\bigg\{\frac{qp}{\hat b}r(t)\bigg\}\exp\bigg\{-\frac{qp}{\hat b}e^{-\hat b T}r_0 - \frac{qp}{\hat b} \int_0^t e^{\hat b(s-T)}\big[\theta(s)+b\xi(s)\big]ds\\
 &&- \frac{qp}{\hat b} \int_0^t e^{\hat b (s-T)}bdW_1(s)  \bigg\}\exp\bigg\{qp\int_0^t r(s)ds \bigg\}\exp\bigg\{qp \int_0^t \sigma_1(s)\pi_1^*(s)dW_1(s)\bigg\}.
\end{eqnarray*}
In addition, we have
\begin{eqnarray*}
&&\exp\bigg\{qp\big[\frac{1}{\hat b}r(t)+\int_0^t r(s)ds\big]\bigg\}\\
&&=\exp\bigg\{qp\big[\frac{r_0}{\hat b}+\frac{1}{\hat b} \int_0^t (\theta(s)+b\xi(s))ds + \frac{1}{\hat b}\int_0^t bdW_1(s)\big]\bigg\}.
\end{eqnarray*}
Then we get
\begin{eqnarray*}
E.q.(4.15)&=& \exp\bigg\{-\frac{qp}{\hat b}e^{-\hat b T}r_0 - \frac{qp}{\hat b}\int_0^t e^{\hat b (s-T)}\big[\theta(s)+ b\xi(s)\big]ds +\frac{qp}{\hat b}r_0 + \frac{qp}{\hat b} \int _0^t \big[\theta(s)+b\xi(s)\big]ds \bigg\}\\
&&\cdot \exp\bigg\{\int_0^t \big[\frac{b}{\hat b}+qp\sigma_1(s)\pi_1^*(s)-\frac{qp}{\hat b}e^{\hat b (s-T)}b\big]dW_1(s)\bigg\}.
\end{eqnarray*}

We can see the first exponential term in the left hand of the equation is deterministic and continuous, on the other hand the second exponential term contain an Ito integral of which integrand is deterministic and continuous. The remainder of the proof is similar to lemma 4.1, we can treat them in the same way, so omit it.

\end{proof}

We call the strategy which satisfy the lemma 4.1 or lemma 4.2 above having the property U. We will see this property is very useful, because it allows the interchange of expected value and limit. Further we should note that the theorem 4.3 below is suitable to both the Ho-lee model and the vasicek model.

Before we begin to prove verification theorem, we give some definition which is similar to \citet{kraft2009optimal}. Define :
\[\Im:=[0,\infty )\times\Gamma_r\times\Gamma_I,\]
where $\Gamma_r$ is the range of the stochastic interest rate $r$, $\Gamma_I$ is the range of the inflation index $I$. Again define
\[\Im_k:=[0,\infty )\times\mathbb{R}\times\mathbb{R}\cap\{z\in\mathbb{R}^3:|z|<k \quad dist(z,\partial\Im)>k^{-1} \}, \quad\quad k\in \mathbb{N}, \]
and
\[Q_k:=[0,T-\frac{1}{k})\times \Im_k ,  \]
where the sets $Q_k$ are not empty for $k\in \mathbb{N} $ with $ k>\frac{1}{T}=:\widetilde{k} $. Without loss of generality, we therefore assume $k>\widetilde{k}$. In addition, let $\theta_k$ be the first exit time of $(t,X(t),r(t),I(t))$ from $Q_k$. Note that for $k\longrightarrow\infty$ we have  $\theta_k\longrightarrow T \quad a.s.$
\\

\begin{thm}
For all $\widetilde\pi=(\pi_1,\pi_2,u)\in\Theta$ we have:
\begin{eqnarray}
E_{t,x,r,I}\big[\frac{1}{p}(X_T)^p\big]\leq G(t,x,r,I).
\end{eqnarray}
Further assume $\widetilde\pi^*=(\pi^*_1,\pi^*_2,u^*)\in\Theta$, if $\widetilde\pi^*=(\pi^*_1,\pi^*_2,u^*)$ has property U, we get
\begin{eqnarray}
E_{t,x,r,I}\big[\frac{1}{p}(X^{\widetilde\pi^*}_T)^p\big]= G(t,x,r,I),
\end{eqnarray}
where $G$ is defined by (3.17) or $(3.23)$
\end{thm}

\begin{proof}
For $ \forall \tau > t $ and all $\widetilde\pi=(\pi_1,\pi_2,u)\in \Theta $, Ito's formula lead to
\begin{eqnarray*}
G(\tau,X_\tau,r_\tau,I_\tau)&=&G(t,x,r,I)+\int_t^\tau\mathcal{A}^{\pi}G(s,X_s^\pi,r_s,I_s)ds+\int_t^\tau u(s)\sigma_3(s)dW_3(s)\\
&&+\int_t^\tau[ X(s)\pi_1(s)\sigma_1(s)+b]dW_1(s)+\int_t^\tau X(s)\pi_2(s)\sigma_2(s)dW_2(s)\\
&&+\int_t^\tau[\bar\sigma_0(s)-X(s)\sigma_0(s)]dW_0(s).
\end{eqnarray*}
For $\widetilde\pi $ is an admissible control, HJB equation (3.3) imply
\[\mathcal{A}^{\widetilde \pi}G(s,X_s^{\widetilde\pi},r_s,I_s)\leq0,\]
consequently
\begin{eqnarray*}
G(\tau,X_\tau,r_\tau,I_\tau)&\leq& G(t,x,r,I)+\int_t^\tau u(s)\sigma_3(s)dW_3(s)+\int_t^\tau[ X(s)\pi_1(s)\sigma_1(s)+b]dW_1(s)\\
&&+\int_t^\tau X(s)\pi_2(s)\sigma_2(s)dW_2(s)+\int_t^\tau[\bar\sigma_0(s)-X(s)\sigma_0(s)]dW_0(s).
\end{eqnarray*}
Because of the property of supermartingale, taking expectation lead to
\[E_{t,x,r,I}[G(\tau,X_\tau,r_\tau,I_\tau)]\leq G(t,x,r,I).\]
Combining with $G(T,x,r,I)=\frac{1}{p}x^p$, we obtain
\[E_{t,x,r,I}[\frac{1}{p}(X_T)^p]\leq G(t,x,r,I).\]
We have prove (4.17).

Next we begin to prove (4.18). Denote $X^*:=X^{\widetilde\pi^*}$, for stopping time $\theta_p$ which is defined before, by Ito's formual
\begin{eqnarray}
\nonumber G(\theta_k,X_{\theta_k},r_{\theta_k},I_{\theta_k})&=&G(t,x,r,I)+\int_t^{\theta_k}\mathcal{A}^{\widetilde\pi}G(s,X_s^*,r_s,I_s)ds\\
\nonumber&&+\int_t^{\theta_k} u(s)\sigma_3(s)dW_3(s)+\int_t^{\theta_k}[ X(s)\pi_1(s)\sigma_1(s)+b]dW_1(s)\\
\nonumber&&+\int_t^{\theta_k} X(s)\pi_2(s)\sigma_2(s)dW_2(s)+\int_t^{\theta_k}[\bar\sigma_0(s)-X(s)\sigma_0(s)]dW_0(s).\\
\end{eqnarray}
From the definition of stopping time $\theta_k$, we know $(X_t,r_t,I_t)$ is bounded on $[0,\theta_k]$. On the other hand, $\sigma_1(t), \sigma_2(t), \sigma_0(t),\bar\sigma_0(t)$, b is deterministic and continuous, hence they are bounded on $[0,\theta_k]$. Therefore, by (3.18) and (3.24), we know $u^*(t), \pi_1^*(t),\pi_2^*(t)$ is bounded on $[0,\theta_k]$ too. Then we have
\begin{eqnarray*}
E_{t,x,r,I}\bigg\{\int_t^{\theta_k} u(s)\sigma_3(s)dW_3(s)+\int_t^{\theta_k}\big[ X(s)\pi_1(s)\sigma_1(s)+b\big]dW_1(s)\\
+\int_t^{\theta_k} X(s)\pi_2(s)\sigma_2(s)dW_2(s)+\int_t^{\theta_k}\big[\bar\sigma_0(s)-X(s)\sigma_0(s)\big]dW_0(s)\bigg\}=0.
\end{eqnarray*}
Additional
\[\mathcal{A}^{\widetilde\pi^*}G(s,X_s^*,r_s,I_s)=0,\]
thus we take expectation for (4.19) obtain
\[E_{t,x,r,I}\bigg\{G(\theta_k,X^*_{\theta_k},r_{\theta_k},I_{\theta_k})\bigg\}=G(t,x,r,I),\]
because sequence $\{G(\theta_k,X_{\theta_k},r_{\theta_k},I_{\theta_k})\}$ is uniformly integrable. It show that
\begin{eqnarray*}
G(t,x,r,I)&=&\lim_{k\rightarrow\infty}E_{t,x,r,I}\bigg\{G(\theta_k,X^*_{\theta_k},r_{\theta_k},I_{\theta_k})\bigg\}\\
&=&E_{t,x,r,I}\bigg\{G(T,X^*(T),r(T),I(T))\bigg\}\\
&=&E\bigg\{\frac{1}{p}(X_T^*)^p|X_t=x,r_t=r,I_t=I\bigg\}.
\end{eqnarray*}

\end{proof}

\section{Numerical analysis}
Now we give some numerical analysis about optimal strategies. For simplicity, we assume the parameters are constant over time interval $t\in [0,T]$, further we take $T=80$, $T_1=120$, $\eta=0.0606$, $b=0.05$, $\rho=-0.06$ and $\beta=0.02$, $\sigma_0=0.01$, $\bar \sigma_0=0.026$, besides, in the comparison of Ho-lee model and Vasicek model showed in figure 3, we let $p=0.5$.

Figure 1 and figure 2 reveal the change of the proportion of the wealth invested in bond with respect to investors who have different attitude toward risk. Parameter $p$, $0<p<1$ in the utility function represent the degree of risk aversion, specifically, the more risk averse the investor is, the larger the parameter is. From figure 1 and figure 2, we know in both models (Ho-lee model and Vasicek model) the insurer should gradually increase the  proportion invested in bond as time elapse. On the other hand, the two figures tell us that an investor who dislike risk will invest less amount of money in bond than the one who like risk.

\begin{center}
\includegraphics[angle=0, width=.6\textwidth]{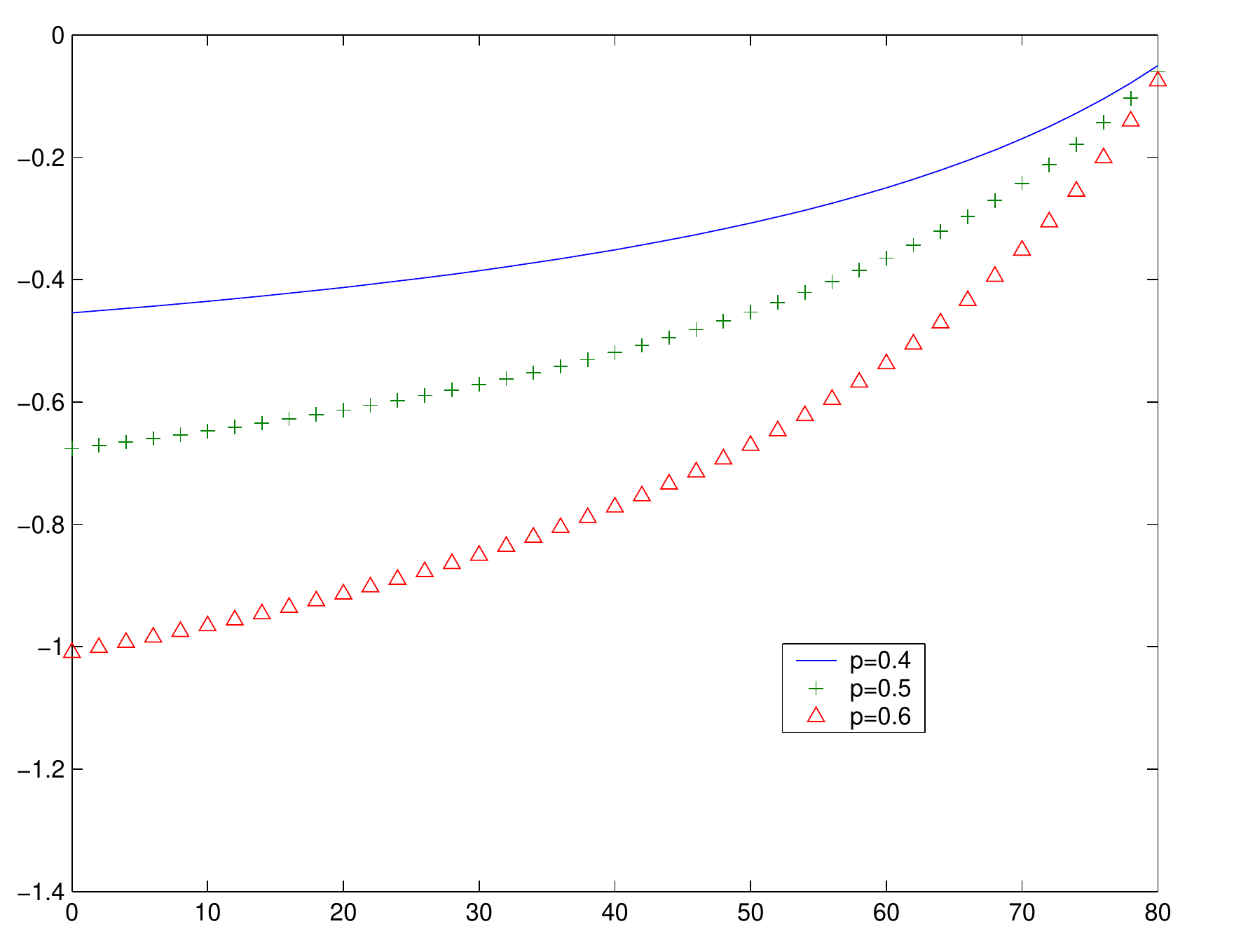}

Figure 1. The case of Ho-Lee model

\end{center}
\begin{center}
\includegraphics[angle=0, width=.6\textwidth]{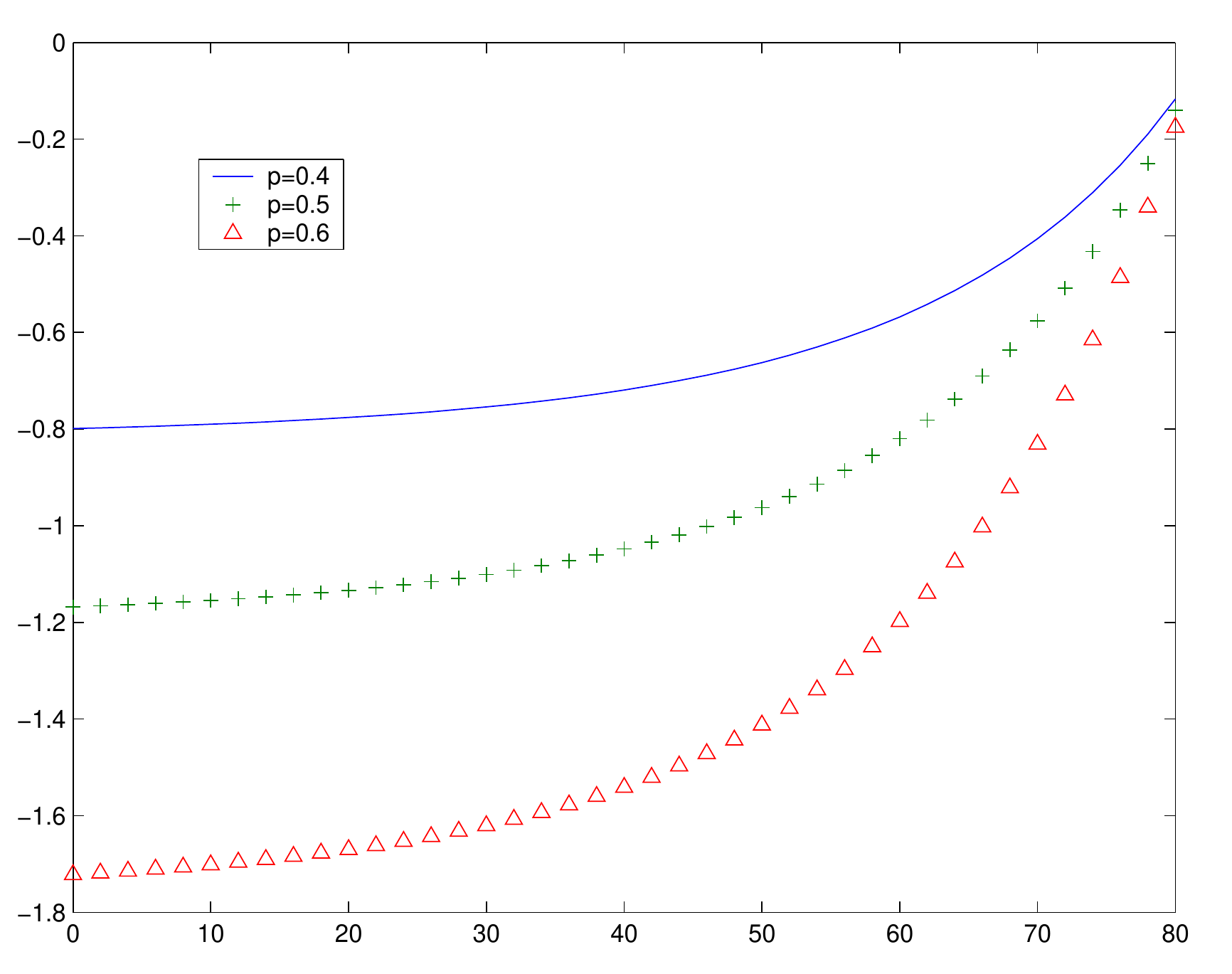}

Figure 2. The case of Vasicek model

\end{center}

Figure 3 shows the comparison of the optimal policies under Ho-lee model and Vasicek model. From image we can see if we use Ho-lee model to describe stochastic interest rate, we will invest more money in bond than the case in which we use Vasicek model to characterize the stochastic interest rate.

\begin{center}
\includegraphics[angle=0, width=.6\textwidth]{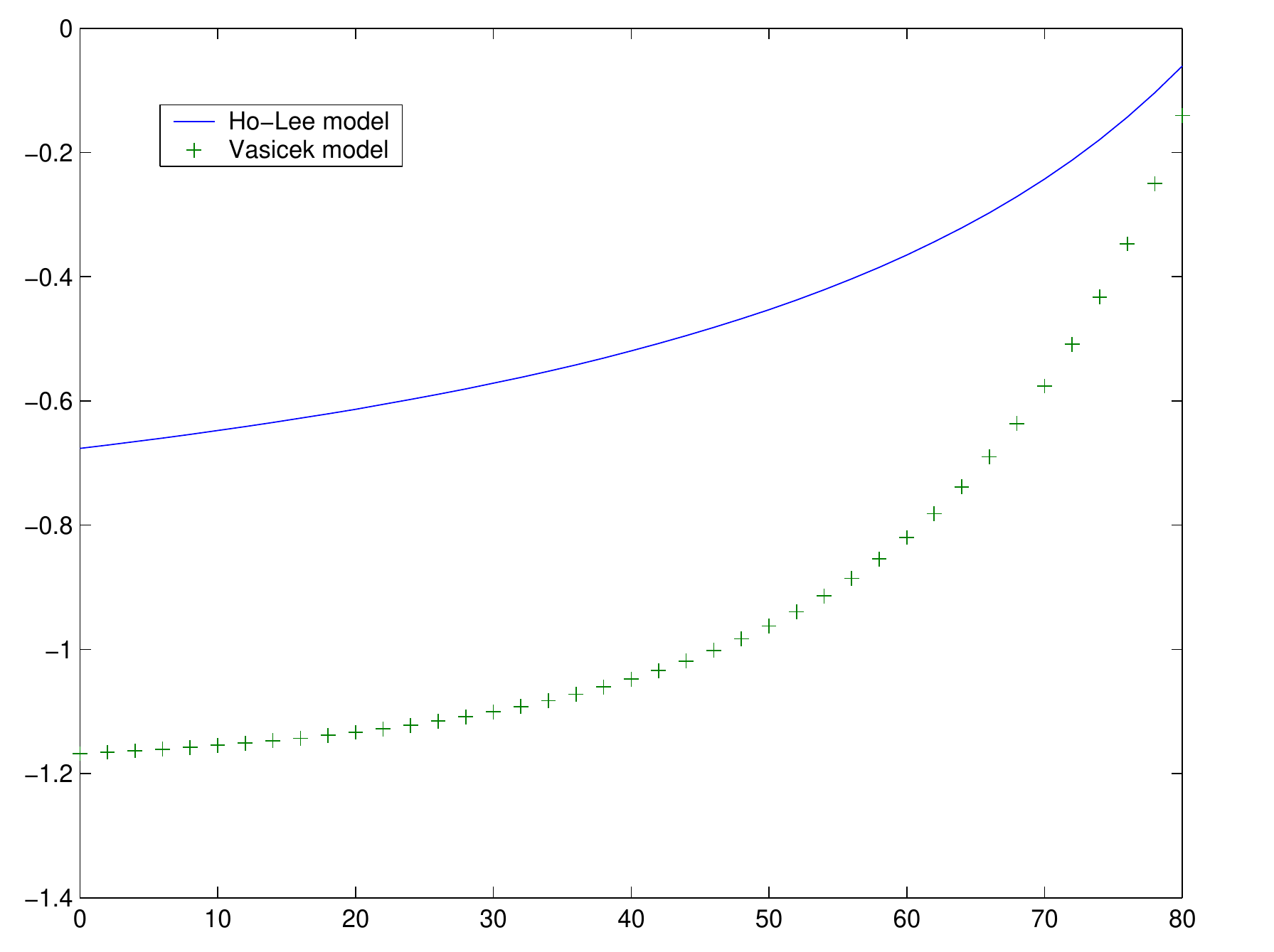}

Figure 3. The comparison of Ho-lee model and Vasicek model in the position of bond

\end{center}

\section{Conclusion}
In this paper, we studied the optimal investment-reinsurance problem under long-term prospective. At first, we established the model which considered the inflation risk and interest risk. Second under the criterion of maximizing the terminal utility function, we obtained the closed-form expression of the optimal strategy for both Ho-lee model and Vasicek model, and by proving corresponding verification theorem, we know the function we obtained is the value function. We also given numerical illustrations on the optimal control strategies.


\end{document}